\documentclass[10pt, a4paper,reqno]{amsart}
\usepackage[margin=1.25in]{geometry}
\usepackage{multirow}
\usepackage{amssymb}
\usepackage{amsmath,amsthm,cite,enumerate,mathtools}
\usepackage{wasysym}
\usepackage{enumitem}
\usepackage{tensor}
\usepackage{booktabs}

\usepackage{etoolbox}
\patchcmd{\subsection}{\textbf}{}{}{}
\patchcmd{\subsection}{-.5em}{.2em}{}{}
\makeatletter
\patchcmd{\@startsection}
  {\@afterindenttrue}
  {\@afterindentfalse}
  {}{}
\makeatother

\usepackage[utf8]{inputenc}

\usepackage{caption}
\newtheorem{theorem}{Theorem}[section]
\newtheorem{lemma}[theorem]{Lemma}
\newtheorem{propos}[theorem]{Proposition}
\newtheorem{corol}[theorem]{Corollary}

\theoremstyle{definition}
\newtheorem{definition}[theorem]{Definition}

\numberwithin{equation}{section}

\newcommand*\tho{\text{\thorn}}
\newcommand*\dho{\text{\dh}}
\DeclareMathOperator{\Tr}{Tr}

\newcommand*\Laplace{\mathop{}\!\mathbin\bigtriangleup}

\newcommand*\Bell{\ensuremath{\boldsymbol\ell}}
\newcommand*\Bn{\ensuremath{\boldsymbol{n}}}
\newcommand*\Bm{\ensuremath{\boldsymbol m}}
\newcommand*\Bk{\ensuremath{\boldsymbol{k}}}
\newcommand*\Bg{\ensuremath{\boldsymbol{g}}}
\newcommand*\BF{\ensuremath{\boldsymbol{F}}}

\newcommand*\Be{\ensuremath{\textbf{e}}}

\newcommand*\Brho{\ensuremath{\boldsymbol\rho}}
\newcommand*\BS{\ensuremath{\textbf{S}}}

\newcommand*\Omegap{\ensuremath{\Omega^{\prime}}}
\newcommand*\omegap{\ensuremath{\omega^{\prime}}}

\newcommand*\de{\ensuremath{\textnormal{d}}}
\newcommand*\bpartial{\ensuremath{\boldsymbol{\partial}}}

\newcommand*\phip{\ensuremath{\varphi^\prime}}

\newcommand{\pp}{{\it pp\,}-}
\def \k {K} 
\def \l {\tilde\Lambda}
\def \Mi {\stackrel{i}{M}}
\def \T {\bigtriangleup  }

\begin{document}

\title{Weyl type N solutions with null electromagnetic fields in the Einstein-Maxwell $p$-form theory}


\author{M. Kuchynka$^{\diamond,\dagger}$,  A. Pravdov\' a$^\dagger$\\
\vspace{0.05cm} \\
{\scriptsize $^\diamond$ Institute of Theoretical Physics, Faculty of Mathematics and Physics, Charles University}, 
{\scriptsize  V Hole\v sovi\v ck\'ach 2, 180 00  Prague 8, Czech Republic}  \vspace{0.04cm} \\
{\scriptsize $^\dagger $ Institute of Mathematics, Czech Academy of Sciences}, \\ 
{\scriptsize \v Zitn\' a 25, 115 67 Prague 1, Czech Republic}   \vspace{0.05cm} \\ 
 {\footnotesize E-mail: \texttt{kuchynkm@gmail.com, pravdova@math.cas.cz}} }



\begin{abstract}
We consider $d$-dimensional solutions to the electrovacuum Einstein-Maxwell equations with the  Weyl tensor of type N 
and a null Maxwell $(p+1)$-form field. We prove that such spacetimes  are necessarily aligned, i.e. 
the Weyl tensor of the corresponding spacetime and the electromagnetic field share the same 
\textit{aligned null direction} (AND). Moreover, this AND is geodetic,  shear-free, non-expanding 
and non-twisting  and hence Einstein-Maxwell equations imply that Weyl type N spacetimes with a null Maxwell $(p+1)$-form field belong to the Kundt class. 
Moreover, these Kundt spacetimes are necessarily $CSI$  and  the $(p+1)$ form is $VSI$.  
Finally, a general coordinate form of  solutions and a reduction of the field equations are discussed. 
\end{abstract}

\maketitle

\section{Introduction and summary}
\label{intro}
We study $d$-dimensional Weyl type N spacetimes with null electromagnetic fields in the context of the generalized Einstein-Maxwell $p$-form theory. 
The theory describes an interaction between a gravitational field $\Bg$ and an electromagnetic field $\BF$ 
and in the electrovacuum case (no sources of the electromagnetic field are present), its action takes the form
\begin{equation}\label{EMakce}
S = \frac{1}{16 \pi} \int \textnormal{d}^d x \sqrt{-g} \left( R - 2\Lambda 
- \tfrac{\kappa_0}{p+1}  F_{a_0 \dots a_{p}} F^{a_0 \dots a_{p}} \right),
\end{equation}
where  $R$ is the Ricci scalar, $\Lambda$ is the cosmological constant, $\kappa_0$ is a coupling constant and
 $\boldsymbol{F}$ is a closed $(p+1)$-form, i.e. 
\begin{equation}\label{EMequations2b}
\nabla_{[a} F_{b_0 \dots b_p ]} = 0.
\end{equation}
 
Varying the action \eqref{EMakce} with respect to the metric $\Bg$ and to the local potential $p$-form $\boldsymbol{A}$ of $\BF$, the least action principle yields a coupled system of the generalized Einstein-Maxwell equations for the pair $(\Bg,\BF)$
\begin{equation}\label{EMequations1}
R_{ab} - \tfrac{1}{2} R g_{ab} + \Lambda g_{ab} = 8 \pi T_{ab}^{EM},
\end{equation}
\begin{equation}\label{EMequations2}
\nabla^a F_{a b_1 \dots b_p} = 0,
\end{equation}
respectively.
Here, $T_{ab}^{EM}$ is the electromagnetic stress-energy tensor associated with the Maxwell field 
$\boldsymbol{F}$, 
\begin{equation}\label{TEM}
 \frac{8 \pi}{\kappa_0} T_{ab}^{EM} \equiv 
F_{a b_1 \dots b_p} F \indices{_b^{b_1 \dots b_p}} 
- \dfrac{1}{2(p+1)}  g_{ab} F_{a_0 \dots a_{p}} F^{a_0 \dots a_{p}}. 
\end{equation}

In the special case of $d=4$ and $p=1$, \eqref{EMakce} reduces to the action of the standard electrovacuum
Einstein-Maxwell theory, which has been extensively studied in the literature. 
We  refer to the system \eqref{EMequations2b} - \eqref{EMequations2}  as the Einstein-Maxwell equations.

Throughout the paper, we restrict ourselves to values $0<p<d-2$ and $d\geq 4$. 
Further, we assume that the Weyl tensor is of algebraic type N in the sense of the null alignment 
classification (see \cite{Milsonetal05,Coleyetal04} or  \cite{review} for a review). This condition can be reformulated in terms of the 
existence of a null vector $\Bell$ such that \cite{MOBelDeb}
\begin{equation}
C_{ab[cd}\ell_{e]}=0.
\end{equation}
The vector $\Bell$ is then referred to as the \textit{Weyl aligned null direction} (Weyl AND or WAND) of the spacetime.  
As in four dimensions, such a spacetime represents a transverse gravitational wave propagating along the null direction $\Bell$ (see the physical interpretation of distinct null frame components of the Weyl tensor carried out in \cite{deviation}). 

Regarding the electromagnetic field $\BF$, we assume that $\BF$ is a (non-vanishing) \textit{null} {(i.e. type N) } form {\cite{Milsonetal05}}, 
i.e. there exists a null vector $\Bk$ such that { \cite{vsiformy}}
\begin{equation}\label{nullconditions}
k^a F_{a b_1 \dots b_p} = 0, \qquad k_{[a}F_{b_0 \dots b_p]} = 0,
\end{equation}
where $\Bk$ is an AND of the Maxwell field 
$\BF$. 
In four-dimensional spacetimes, distinguished geometrical role of null electromagnetic two-forms satisfying 
the source-free Maxwell equations is well known due to the Mariot-Robinson theorem (see theorem 7.4 in \cite{stephani}).  
It states that \textit{a four-dimensional spacetime admits a shear-free geodetic null congruence if and only if it 
admits  a null Maxwell 2-form (a test field) satisfying the Maxwell equations}.  In higher dimensions, there is a partial generalization of the 
theorem  for  Maxwell $(p+1)$-form fields 
(see lemmas 3 and 4 in \cite{ghpclanek}). 

In four dimensions, it is known that Petrov type N Einstein-Maxwell fields are either non-aligned and non-null or aligned (null or non-null) and 
 Kundt (see tables 38.9 and 38.10 in \cite{stephani} and  theorems 7.4, 28.4 and 30.3 therein) and 
any type III/N pure-radiational metric is necessarily aligned 
(see  theorem 3 of \cite{wils}).
Thus, in four dimensions, it follows  that \textit{any Petrov type N null Einstein-Maxwell field is aligned and Kundt}. Here, we extend this result to any dimension $d \geq 4$ and any 
electromagnetic $(p+1)$-form field such that $0<p<d-2$.  

As a consequence of our previous result \cite{NNspacetimes} it follows that \textit{a spacetime  
corresponding to a Weyl type N solution of the Einstein-Maxwell equations with a null Maxwell $(p+1)$-form $\BF$ 
is necessarily aligned with $\BF$}. In other words, the null vectors $\Bk$ and $\Bell$ defined above are parallel. 
Proposition \ref{alignedKundt} then states that \textit{the common AND $\Bell$ of the Weyl tensor and of the Maxwell field $\BF$ is tangent to a shear-free,  expansion-free and non-twisting geodetic null congruence}, i.e. \textit{the spacetime 
belongs to the Kundt class}. 
\begin{propos}\label{alignedKundt}
All Weyl type N solutions of the Einstein-Maxwell equations with a non-vanishing null $(p+1)$-form field are aligned and Kundt. 
\end{propos}

Further, applying results of appendix \ref{apendixB},  
where it is shown that all Weyl type N Kundt spacetimes with a type N traceless Ricci tensor are CSI and an arbitrary covariant derivative of 
their Riemann tensor is even VSI,
it follows that: 
\begin{propos}\label{CSIaVSI}
Any Weyl type N solution of the Einstein-Maxwell equations with a non-vanishing null electromagnetic field consists of a 
$CSI$ spacetime and  a $VSI$ $(p+1)$-form.\footnote{$\boldsymbol{T}$ is a $CSI$ (constant scalar invariant) 
tensor if all scalar polynomial invariants constructed from $\boldsymbol{T}$ and its covariant derivatives 
of arbitrary order are constant. If, moreover, all these scalar invariants vanish, we say that $\boldsymbol{T}$ 
is a VSI (vanishing scalar invariant) tensor. $CSI$/$VSI$ spacetime is then a spacetime whose Riemann 
tensor is $CSI$/$VSI$, see \cite{csiclanek,ColHerPel09b,ColHerPel10} and \cite{VSI4d,HDVSI2,vsiclanek}, respectively (see also \cite{review} and reference therein).}  For $\Lambda = 0$, the corresponding spacetime is $VSI$. 
\end{propos}
As a consequence, such solutions possess a certain universality property \cite{vsiformy},\cite{MPuniv}.

The paper is organized as follows. In the next  section, we review some basic notions and notation employed throughout the paper and briefly discuss the structure of the field equations following from \eqref{EMakce} in the case of the null Maxwell field (section \ref{notationsekce}). 
 Section \ref{properties} contains proofs of  propositions \ref{alignedKundt} and \ref{CSIaVSI}, in which more general results of  appendix \ref{apendixB} are employed. 
In section \ref{generalform}, a general coordinate form of both the metric and the Maxwell field in adapted coordinates is obtained. 
Subsequently, a reduction of the Einstein-Maxwell equations to a set of equations for quantities emerging in the local description of $\Bg$ and $\BF$ follows. A  class of solutions { with a covariantly constant WAND},  \pp-waves, is briefly discussed.
At the end of the paper, appendix \ref{apendixB} devoted to the study of scalar curvature invariants in general Weyl type N Kundt spacetimes of traceless Ricci type N is included.

\section{Preliminaries}
\label{notationsekce}
For purposes of this paper, let us recall the notation, algebraic classification and NP and GHP formalisms as were introduced in
\cite{Coleyetal04,Milsonetal05,ricciclanek,ghpclanek}
and summarized in \cite{review}.
In a $d$-dimensional spacetime, we shall consider a local \textit{null} frame $\{ \textbf{e}_{(a)} \}  \equiv \{ \Bell, \Bn, \Bm_{(i)} \}$ 
with two null vector fields $\Bell,\ \Bn$ and $d-2$ spacelike vector fields $\{ \Bm_{(i)} \}$ such that they satisfy the following orthogonality relations 
\begin{equation}\label{ortogonality}
\ell_a  n^a = 1, \qquad 
\ell_a \ell^a = n_a n^a = \ell_a m_{(i)}^a = n_a m_{(i)}^a = 0, \qquad
 m_a^{(i)}  m_{(j)}^a = \delta_j^i .
\end{equation}
Let us note that relations \eqref{ortogonality} are invariant under Lorentz transformations of $\{ e_{(a)} \}$. 

\subsection{Ricci rotation coefficients and directional derivatives}
Now, let us introduce a few basic objects of the higher-dimensional Newman-Penrose formalism. 
Given a null frame $\{ \Bell,\ \Bn,\ \Bm_{(i)} \}$,  consider covariant derivatives of the individual frame vectors
\begin{equation}
L_{ab} \equiv \nabla_b \ell_a, \qquad N_{ab} \equiv \nabla_b n_a, \qquad 
\overset{i}{M}_{ab} \equiv \nabla_b m_{a}^{(i)}.  
\end{equation}
Projecting these derivatives onto the null frame, one obtains a set of scalars $L_{(a)(b)}$, $N_{(a)(b)}$ and 
$\overset{i}{M}_{(a)(b)}$, the so-called \textit{Ricci rotation coefficients}. We  omit the parenthesis, whenever it is clear that the corresponding quantities are projections of a tensor onto the null frame. Due to the orthogonality relations 
\eqref{ortogonality} satisfied by the null frame vectors, the corresponding rotation coefficients fulfill the following equalities
\begin{equation}
N_{0a} + L_{1a} = 0, \qquad \overset{i}{M}_{0a} +  L_{ia} = 0, \qquad \overset{i}{M}_{1a} + N_{ia} = 0, 
\qquad \overset{i}{M}_{ja} + \overset{j}{M}_{ia} = 0, 
\end{equation}
\begin{equation}
L_{0a} = N_{1a} = \overset{i}{M}_{ia} = 0.
\end{equation}
For transformation properties of the Ricci rotation coefficients under the Lorentz transformations, see \cite{ricciclanek}.
Lastly, let us denote the null frame directional derivatives as 
\begin{equation}
D \equiv \ell^a \nabla_a, \qquad \Delta \equiv n^a \nabla_a, \qquad \delta_i \equiv m_{(i)}^a \nabla_a.
\end{equation}
Using these, the action of the covariant derivative $\nabla$ can be decomposed in the following way 
\begin{equation}\label{covariantDecomposition}
\nabla_a = n_a D + m_a^{(i)} \delta_i + \ell_a \Delta.
\end{equation}

\subsection{Optical matrix and optical scalars}
When studying  geometrical properties of null congruences, it is usually convenient to employ an adapted null 
frame with $\Be_{(0)}$ being the null vector tangent to the congruence. Then, some information on geometry 
of the congruence is encoded in $L_{ab}$. For example, the congruence is geodetic if and only if 
$\kappa_i \equiv L_{i0}$ vanishes. In such a case, the scalar $L_{10}$ corresponds to an affine parametrization of the congruence - it is affinely parametrized if and only if $L_{10} = 0$. 

For a geodetic congruence, some other geometrical properties are encoded in the optical matrix $\rho_{ij} \equiv L_{ij}$. 
Consider the following decomposition of $\rho_{ij}$
\begin{equation}
\rho_{ij} = \sigma_{ij} + \theta \delta_{ij} + A_{ij},
\end{equation}
 where $\sigma_{ij}$ denotes the shear part, $\theta \delta_{ij}$ is the expansion part and $A_{ij}$ is 
the antisymmetric part of $\rho_{ij}$. 
Using these particular matrices, the following \textit{optical scalars} can be defined: 
 the \textit{expansion} $\theta$, \textit{shear} $\sigma^2 \equiv \sigma_{ij}\sigma_{ij}$ 
and \textit{twist} $\omega^2 \equiv A_{ij} A_{ij}$ of the corresponding congruence.

\subsection{Weyl type N spacetimes and null electromagnetic fields}
Recall the conditions \eqref{nullconditions} for $\BF$ to be a null form. Employing an adapted frame $\{ \Bk, \Bn, \Bm_{(i)}\}$ with $\Bk$ being the null vector emerging in \eqref{nullconditions}, the only non-vanishing independent null frame components of $\boldsymbol{F}$ are the boost weight $(-1)$ quantities
\begin{equation}\label{nullF}
\varphi_{k_1 \dots k_p}^\prime \equiv F_{a b \dots c} n^{a} m_{(k_1)}^b \dots m_{(k_p)}^c, 
\end{equation}
i.e. $\BF$ is of type N with the aligned null direction (AND) $\Bk$ (see e.g. section 3 of \cite{ghpclanek}).  

For a Weyl type N spacetime, the only non-trivial independent frame components of the Weyl tensor $\boldsymbol{C}$ in an adapted frame with $\Be_{(0)} \equiv \Bell$ being the corresponding WAND are 
\begin{equation}
\Omegap_{ij} \equiv C_{abcd} n^a m_{(i)}^b n^c m_{(j)}^d.
\end{equation}
If also the components $\Omegap_{ij}$ vanish, the Weyl tensor is said to be of algebraic type O and the spacetime is conformally flat.

\subsection{Einstein equations for null electromagnetic fields} 
For a null form $\boldsymbol{F}$ { (\ref{nullconditions})}, the scalar invariant $F_{a_0 \dots a_{p}} F^{a_0 \dots a_{p}}$ in \eqref{TEM} vanishes and $T_{ab}^{EM}$ is of type N (with the same AND { $\Bk$}). 
In an adapted frame { $\{ \Bk, \Bn, \Bm_{(i)}\}$},  its only non-vanishing null frame component reads
\begin{equation}\label{EMtensor}
T_{11}^{EM} = \frac{\kappa_0}{8\pi} F_{1 a_1 \dots a_p} F \indices{^{0 a_1 \dots a_p}}.
\end{equation}
Thus, denoting $\tilde \Lambda \equiv 2 \Lambda / (d-2)$, the field equation \eqref{EMequations1} takes the form
\begin{equation}\label{EMequation1i}
R_{ab} =  \tilde \Lambda g_{ab} + 8\pi T_{11}^{EM} { k_a k_b},
\end{equation}
or, equivalently, in the GHP formalism  (see table 3 of \cite{ghpclanek})
\begin{equation}\label{EMequation1ghp}
\omega = 0, \qquad \psi_i = 0, \qquad \phi = \tilde \Lambda, \qquad 
\phi_{ij} = \tilde \Lambda \delta_{ij}, \qquad 
\psi_i^\prime = 0, \qquad
\omegap = 8 \pi T_{11}^{EM}.
\end{equation}
Thus, the corresponding solution is of traceless Ricci type N and the multiple Ricci AND is precisely the multiple AND { $\Bk$} of the Maxwell field $\boldsymbol{F}$. 
Moreover, due to the alignment result for spacetimes of Weyl and traceless Ricci type N (see proposition 2.1 of \cite{NNspacetimes}), we obtain that \textit{the Maxwell field $\boldsymbol{F}$ is aligned with the Weyl tensor}, i.e. $\Bk$ is also a quadruple WAND in the spacetime, i.e. $\Bk \propto \Bell$ and without loss of generality, one can set $\Bk = \Bell$. Furthermore, $\Bell$ \textit{is geodetic}  - this follows independently either from the contracted Bianchi identity or from the source-free Maxwell equations for a null Maxwell (test) field (see \cite{MO07} for the 2-form case and \cite{ghpclanek} for an arbitrary $p$).

\subsection{Maxwell equations for null electromagnetic fields}
Employing an adapted null frame $\{ \Bell, \Bn, \Bm_{(i)}\}$, the source-free GHP Maxwell equations \eqref{EMequations2b}, \eqref{EMequations2} for a null Maxwell field $\BF$ with a geodetic AND $\Bell$ reduce to (see section 3 of \cite{ghpclanek}):
\begin{align}
\left( p \rho_{[k_1 |i} + p \rho_{i[k_1|} - \rho \delta_{[k_1|i} \right) \phip_{i|k_2 \dots k_p]} &= 0, \label{3.5} \\
\phip_{[k_1 \dots k_p} \rho_{k_{p+1} k_{p+2}]} &= 0,\label{3.6} \\
\rho_{[ij]} \phip_{ij k_1 \dots k_{p-2}} &= 0, \label{3.7} \\
\dho_{i} \phip_{i k_1 \dots k_{p-1}} &= \tau_{i} \phip_{i k_1 \dots k_{p-1}}, \label{3.3'} \\
\dho_{[k_1} \phip_{k_2 \dots k_{p+1}]} &= \tau_{[k_1} \phip_{k_2 \dots k_{p+1}]}, \label{3.4'} \\
2 \tho \phip_{k_1 \dots k_p} &= \left(
p \rho_{[k_1 | i} - p \rho_{i [ k_1|} - \rho \delta_{[ k_1 | i}
 \right) \phip_{i | k_2 \dots k_{p}]}. \label{3.5'}
\end{align}
For $p=1$, equation \eqref{3.7} does not appear. Also note that for $p>d-4$, equation \eqref{3.6} is identically satisfied.

\section{Proofs of the main results}
\label{properties}

\begin{proof}[Proof of proposition \ref{alignedKundt}]
We have already argued that $\Bell$ is the common geodetic AND of the Weyl tensor and the corresponding Maxwell field. 
Now, we shall prove that its optical matrix $\Brho$ vanishes. 
In order to do that, we start from the canonical form of the optical matrix corresponding to the WAND in spacetimes of Weyl and traceless Ricci type N (given by equation (1.3) of \cite{NNspacetimes}) and argue that it is compatible with restrictions following from the Maxwell equations (see lemmas 3 and 4 in \cite{ghpclanek}) only if $\Brho = \boldsymbol{0}$. 
For simplicity, let us assume that $p>1$. The $p=1$ case 
can be easily proved in a similar way using the stronger result of lemma 4 of \cite{ghpclanek}. 

Let $\mathcal{F}$ be a frame in which $\Brho$ takes the canonical form and let 
$S_{ij} \equiv \rho_{(ij)}$, $A_{ij} \equiv \rho_{[ij]}$, $\rho \equiv \Tr \Brho$. 
The canonical form of $\Brho$ in spacetimes of Weyl and traceless Ricci type N reads 
$\Brho =  \textnormal{diag}(\mathcal{L}, \boldsymbol{0}, \dots,\boldsymbol{0})$ with the only possibly non-vanishing $2 \times 2$ block
\begin{equation}\label{blockL}
\mathcal{L} = s 
\renewcommand*{\arraystretch}{1.25}
 \begin{bmatrix}
  1& {a\ } \\
  -a\ & {b\ }  
 \end{bmatrix},
\end{equation}
where $b \neq  1$, otherwise the spacetime is Einstein, i.e. $\omegap = 0$. 

First, we  show that $\Brho$ is traceless, i.e. $\rho = 0$.
According to lemma 3 of \cite{ghpclanek}, the eigenvalues $\{ S_i \}$ of $\BS$ can be rearranged such that 
\begin{equation}\label{twosums}
\sum_{i=2}^{p+1} S_i = \frac{\rho}{2} = \sum_{j=p+2}^{d-1} S_j.
\end{equation}
Assume that $\rho$ is non-vanishing. Then each of these two sums has to contain at least one 
non-vanishing element of $S_i$. 
In our case, the only non-vanishing eigenvalues of $\BS$ are $s$ and $sb$. Hence, 
each sum in \eqref{twosums} has to contain exactly one of these two elements. 
Since the rest of the eigenvalues of $\BS$ vanishes, these sums reduce to $s = \rho/2$ and $sb = \rho/2$. 
In particular, one has that $b=1$. This contradicts our assumption of the presence of a non-trivial Maxwell field $\boldsymbol{F}$ in the spacetime. Hence $\rho = 0$.

Now, we  prove that this already implies $\Brho = \boldsymbol{0}$. 
From $\rho = 0$, we have that $s=0$ or $b=-1$. Of course, if $s=0$, we are done.   
Hence, let us suppose that $s \neq 0$ and $b=-1$ instead. Note that, at the moment, $\Brho$ can be non-trivial only if both the shear and the twist of the corresponding geodetic null congruences are non-vanishing, otherwise the Sachs equation implies $\Brho = \boldsymbol{0}$ (see (15b) in \cite{ricciclanek}). Hence both $s$ and $a$ in \eqref{blockL} are non-vanishing. 
Let us define auxiliary GHP scalars $\boldsymbol{T}$, $\boldsymbol{U}$ and $\boldsymbol{V}$ as 
\begin{align}
T_{k_1  \dots k_p} &\equiv S_{[k_1| i} \varphi_{i|k_2 \dots k_p]}^\prime 
= S_{[k_1| 2} \varphi_{2|k_2 \dots k_p]}^\prime +  S_{[k_1| 3} \varphi_{3|k_2 \dots k_p]}^\prime, \\
U_{k_1 \dots k_{p+2}} &\equiv \varphi_{[k_1 \dots k_p}^\prime A_{k_{p+1} k_{p+2}]}, \\
V_{k_1 \dots k_{p-2}} &\equiv A_{ij}\varphi_{ij k_1 \dots k_{p-2}}^\prime 
= 2 s a \varphi_{23 k_1 \dots k_{p-2}}^\prime  \label{V}
\end{align} 
($\boldsymbol{V}$ is defined only for $p>1$). The Maxwell equations \eqref{3.5}, \eqref{3.6}, and \eqref{3.7} then read  
\begin{align}\label{TUV}
\boldsymbol{T} = \boldsymbol{0}, \qquad
\boldsymbol{U} = \boldsymbol{0}, \qquad
\boldsymbol{V} =\boldsymbol{0} .
\end{align} 
In view of \eqref{V}, the last equation of \eqref{TUV} immediately implies $ \varphi_{23 k_1 \dots k_{p-2}}^\prime = 0$ for 
every $k_1, \dots, k_{p-2} >3$.  
Let us further examine $\boldsymbol{T}$ and $\boldsymbol{U}$ for a particular choice of indices.  
For every $k_1, \dots, k_p >3$, we have the following
\begin{align}
s \varphi_{2 k_2 \dots k_p}^\prime & \propto T_{2 k_2  \dots k_p}  , \label{T1} \\ 
s \varphi_{3 k_2 \dots k_p}^\prime &\propto T_{3 k_2  \dots k_p}  , \label{T2} \\
s a \varphi_{ k_1 \dots k_p}^\prime & \propto U_{2 3 k_1 \dots k_p}.
\end{align} 
Thus, according to the first and the second equation of \eqref{TUV}, also $\varphi_{2 k_2 \dots k_p}^\prime $, $\varphi_{3 k_2 \dots k_p}^\prime $ and 
$ \varphi_{ k_1 \dots k_p}^\prime$ vanish. Hence
$\boldsymbol{\varphi}^\prime$ vanishes completely. Since the Maxwell form $\boldsymbol{F}$ is non-vanishing, we arrive at a contradiction. Thus, $\boldsymbol{\rho}$ has to be zero and the spacetime is Kundt. 
\end{proof}
\noindent 
Moreover, since the corresponding spacetime is of the Weyl type N with the Kundt AND $\Bell$ and its Ricci tensor takes the form 
\eqref{EMequation1i}, we conclude that it is a degenerate Kundt spacetime (see sec. 7.1.2 in \cite{review} and references therein).  

\begin{proof}[Proof of proposition \ref{CSIaVSI}]
Suppose that $(\Bg,\BF)$ is a solution of the Einstein-Maxwell { equations} satisfying our assumptions. The $CSI$ property of $\Bg$ follows immediately from theorem \ref{alignedKundt} and corollary \ref{CSIKundt}. 
In the case of $\Lambda = 0$, the spacetime 
 is of Weyl and Ricci type N. Hence its Riemann tensor is also of type N and  
{ hence the spacetime is $VSI$} (see theorem 1 of \cite{HDVSI2}). 

To prove that $\BF$  is $VSI$, recall that the spacetime 
is a degenerate Kundt spacetime and $\BF$ is a closed form of type N.  
Then, employing the characterization of VSI electromagnetic $(p+1)$-forms (theorem 1.5 of \cite{vsiformy}), we arrive at the desired result. 
\end{proof}

Theorems \ref{alignedKundt} and \ref{CSIaVSI} suggest that such solutions may be of interest also in various generalized theories. 
In particular, $\BF$ corresponding to any solution $(\Bg, \BF)$ under consideration is a \textit{universal} \cite{MPuniv} test solution to generalized electrodynamics on the fixed background {metric} $\boldsymbol{g}$,  i.e. it is a solution to any electrodynamics with the field equations of the form 
\begin{equation}\label{generalelectro}
\de \boldsymbol{F} = \boldsymbol{0}, \qquad \ast\ \de\ast\tilde{\boldsymbol{F}} = \boldsymbol{0},
\end{equation}
where $\ast$ denotes the Hodge dual in $\boldsymbol{g}$ and $\tilde{\boldsymbol{F}}$ is any $(p+1)$-form constructed as a polynomial of $\boldsymbol{F}$ and its covariant derivatives of an arbitrary order (see section 2.4 of \cite{vsiformy}).\footnote{For $\tilde{\boldsymbol{F}}=\boldsymbol{F}$, 
 equations \eqref{generalelectro} reduce to the Maxwell equations.}  
Thus, $(\boldsymbol{g},\boldsymbol{F})$ is also a solution to the Einstein equations \eqref{EMequations1} coupled with the electrodynamics \eqref{generalelectro}.

\section{General form of the solution}
\label{generalform}
In this section, we discuss the general form of a Weyl type N metric $\boldsymbol{g}$ and a null Maxwell $(p+1)$-form field $\boldsymbol{F}$ satisfying the Einstein Maxwell equations \eqref{EMequations2b}--\eqref{EMequations2}.  

Since ($\Bg$, \BF) is a degenerate Kundt spacetime with a VSI form, the discussion on the local form of a solution in sections 2 and 3 of the paper \cite{vsiformy} applies. 
In particular, $\Bg$ and $\BF$ take the local coordinate form (9) and (13) of \cite{vsiformy}, respectively. 
These are then subject to the field equations (15), (21), (22), (24), and (25) of the corresponding paper. 
In particular, negative boost weight field equations $\omega = 0$, $\psi_i = 0$ are automatically satisfied. 

However, such a spacetime is, in general, of the Weyl type II. Thus, in order to obtain a solution of the Weyl type N, conditions on an algebraic type of the Weyl tensor need to be imposed. 
This can be achieved by employing conditions II(a) - III(b) of \cite{Kundtclass}, where an explicit algebraic classification of the general Kundt line element was carried out. 
Doing so, the local form of both $\Bg$ and the field equations further reduces. 

It is convenient to discuss the coordinate form of the metric first. Then, we briefly discuss the general coordinate form of the electromagnetic field.

As has already been said, the local form of the metric in adapted Kundt coordinates $\{ r,u,x^\alpha \}$ with $r$ being an affine parameter of the multiple WAND $\Bell \equiv \bpartial_r$ is given by (9) of \cite{vsiformy}. 
Employing Weyl type N conditions II(a) - III(b) of  \cite{Kundtclass},  { the } local form of $\Bg$ further reduces to 
\begin{equation}\label{localmetric}
\de s^2 = 2 H(r,u,x)\de u^2 + 2\de u \de r + 2 W_\alpha(r,u,x) \de x^\alpha \de u + g_{\alpha \beta}(u,x) \de x^\alpha \de x^\beta, 
\end{equation}
where $\boldsymbol{g}^\perp \equiv g_{\alpha \beta} \de x^\alpha \de x^\beta$ is a Riemannian metric on a $(d-2)$-dimensional transverse space spanned by $\{ x^\alpha \}$, which is of constant sectional curvature $\k$ depending on the cosmological constant and the dimension of the spacetime as 
\begin{equation}\label{ccmetric}
\k = \frac{\tilde \Lambda}{d-1} = \frac{2 \Lambda}{(d-1)(d-2)}
\end{equation}
and the metric functions $H$ and $W_\alpha$ take the form
\begin{equation}
W_{\alpha}(r,u,x) = r W_{\alpha}^{(1)}(u,x) + W_{\alpha}^{(0)}(u,x),
\end{equation}
\begin{equation}
H (r,u,x) = r^2 H^{(2)}(u,x) + r H^{(1)}(u,x) + H^{(0)}(u,x).
\end{equation}
As a consequence of the fact that $\boldsymbol{g}^\perp$ is a constant curvature metric with $\k$ given by \eqref{ccmetric}, 
the field equations (22), (21) of \cite{vsiformy} corresponding to the boost weight zero equations 
$\phi = \tilde \Lambda$, $\phi_{ij} = \tilde \Lambda \delta_{ij}$ reduce to 
\begin{equation}\label{EFE1}
2 H^{(2)} = \frac{1}{4}  W_{\alpha}^{(1)} W^{(1) \alpha} + \k,
\end{equation}
\begin{equation}\label{EFE2}
 W_{(\alpha || \beta)}^{(1)} - \frac{1}{2} W_{\alpha}^{(1)}W_{\beta}^{(1)} = 2 \k g_{\alpha \beta}, 
\end{equation}
respectively. Here, $||$ denotes the covariant derivative in the transverse space with the metric $\boldsymbol{g}^\perp$. 
The contracted Bianchi identity (see (23) of \cite{vsiformy}) for the transverse metric $\Bg^{\perp}$ 
is automatically satisfied provided 
Weyl type condition (IId) of \cite{Kundtclass} holds
\begin{equation}\label{WC1}
W_{[\alpha || \beta]}^{(1)} =0. 
\end{equation}
Further, employing condition III(a) of \cite{Kundtclass}, equation (24) of \cite{vsiformy} corresponding to the boost weight $(-1)$ field equation $\psi_i^\prime = 0$ reads
\begin{equation}\label{EFE3}
2H_{,\alpha}^{(1)} = W_{\alpha,u}^{(1)} - W_{[\alpha || \beta]}^{(0)} W^{(1) \beta}
 + \frac{1}{2} g_{\alpha \beta,u} W^{(1)\beta} + \frac{1}{2} W_{\beta}^{(1)}W^{(0)\beta} W_{\alpha}^{(1)} 
+ 2\k W_{\alpha}^{(0)}.  
\end{equation} 
 Conditions III(a), III(b) of \cite{Kundtclass} reduce to 
\begin{equation}\label{WC2}
H_{, \alpha}^{(2)} - H^{(2)} W_{\alpha}^{(1)} = 0, 
\end{equation}
\begin{equation}\label{WC3}
{W_{[\alpha|| \beta] || \gamma}^{(0)}} = 
 \frac{1}{2}\left(  W_{\gamma \beta}^{(0)} W_{\alpha}^{(1)}  - W_{\gamma \alpha}^{(0)} W_{\beta}^{(1)}  \right)  
- \k \left(  g_{\gamma \beta} W_{\alpha}^{(0)}  - g_{\gamma  \alpha} W_{\beta}^{(0)} \right)
- g_{\gamma [ \beta , u || \alpha ]},
\end{equation}
respectively, where the auxiliary geometrical quantity $W_{\alpha \beta}^{(0)}$ is defined as 
\begin{equation}
W_{\alpha \beta}^{(0)} \equiv W_{(\alpha || \beta)}^{(0)} - \frac{1}{2} g_{\alpha \beta, u}.
\end{equation}
Finally, equation (25) { of \cite{vsiformy}} corresponding to the boost weight $(-2)$ field equation $\omegap = \kappa_0 f_{\alpha \dots \beta} f^{\alpha \dots \beta}$ simplifies to a slightly simpler form
\begin{align}\label{EFE4}
\begin{split}
 \Laplace_{LB} H^{(0)} &+ W^{(1)\alpha} H_{,\alpha}^{(0)} + \left( 4H^{(2)} +2(d-3)\k \right)H^{(0)} = 
2 W_{\alpha}^{(0)} W^{(0)\alpha} H^{(2)} \\
&+ g^{\alpha \beta} W_{\alpha \beta}^{(0)} H^{(1)}  + W_{[\alpha || \beta]}^{0} W^{(0) [\alpha || \beta]} + {W_{\alpha,u}^{(0)}} \indices{^{|| \alpha}}  
- W_{[\alpha || \beta]}^{(0)} W^{(1)\alpha}W^{(0)\beta} \\
&+ \frac{1}{2} g_{\alpha \beta,u} \left( W^{(1)\alpha} W^{(0) \beta} 
- \frac{1}{2}{g^{\alpha \beta}} \indices{_{,u}}  \right) - \frac{1}{2} g^{\alpha \beta} g_{\alpha \beta,uu}  - \kappa_0  f_{\alpha \dots \beta} f^{\alpha \dots \beta},
\end{split}
\end{align} 
where $\Laplace_{LB}$  denotes the Laplace-Beltrami operator in the transverse space, i.e. in coordinates, it acts on a scalar as  $\Laplace_{LB} f \equiv g^{\alpha \beta} f_{,\alpha || \beta}$. 

The coordinate form \eqref{localmetric} of the metric is preserved under the coordinate transformations (81) of \cite{generalKundt}. 
In particular,  since the spacetime 
is a $CSI$ spacetime, without loss of generality one can assume that components $g_{\alpha \beta}$ are independent of $u$ (see theorem 4.1 of \cite{csiclanek}). Hence, in the case of a $VSI$ spacetime ($\Lambda = 0$), the coordinates can be chosen such that $g_{\alpha \beta}= \delta_{\alpha \beta}$ and \eqref{localmetric} reduces to the standard form of the VSI line element (8) of \cite{vsiclanek} for which most of the above equations simplify.

Since the {metric} $\boldsymbol{g}$ is a degenerate Kundt {metric} and the Maxwell field $\boldsymbol{F}$ is a closed form of type N, $\BF$ takes the coordinate form  (see (13) of \cite{vsiformy})
\begin{equation}\label{solutionform}
\boldsymbol{F} = \frac{1}{p!} f_{\alpha_1 \dots \alpha_p}(u,x^\alpha) \de u \wedge \de x^{\alpha_1} \wedge \dots \wedge \de x^{\alpha_p}
\end{equation}
in coordinates $(r,u,x^\alpha)$ adopted in the previous section. Here, $f_{\alpha_1 \dots \alpha_p} \equiv F_{u \alpha_1 \dots \alpha_p}$. 
For $\BF$ in the form \eqref{solutionform}, { the} Maxwell equations \eqref{EMequations2b} and \eqref{EMequations2} reduce to the effective Maxwell equations
\begin{equation}\label{EFE5}
 \big( \sqrt{g^\perp} f^{\beta \alpha_2 \dots \alpha_p}  \big)_{,\beta} = 0, \qquad
f_{[\alpha_1 \dots \alpha_p, \beta]} = 0,
\end{equation}
for a $p$-form $\boldsymbol{f}$ in the $(d-2)$-dimensional transverse Riemannian space (see section 2.2 of \cite{vsiformy}).

In \cite{vsiformy}, it was pointed out that, under suitable conditions, 
solutions $(\Bg,\BF)$ of the Einstein-Maxwell equations with a $VSI$ metric $\Bg$ and a VSI Maxwell field $\BF$ are  \textit{universal} (in the sense of section 3.2 in \cite{vsiformy}) and thus also simultaneously solve various Einstein-generalized Maxwell theories. 
Weyl type N universal solutions considered in \cite{vsiformy} are necessarily $VSI$ $pp$-waves,
 i.e. spacetimes admitting  { a} covariantly constant null vector field
  (see \cite{review,vsiformy} and references therein),  for which the field equations  simplify
	considerably.


Let us conclude with a brief summary of section \ref{generalform}.
Any Weyl type N spacetime with a metric $\boldsymbol{g}$ 
corresponding to a solution 
of the electrovacuum  Einstein-Maxwell equations with a null Maxwell field $\boldsymbol{F}$ can be expressed in the 
coordinate form \eqref{localmetric} while the null Maxwell form $\boldsymbol{F}$ takes the coordinate form \eqref{solutionform}.
 The transverse metric $\boldsymbol{g}^\perp$ in the line element \eqref{localmetric} is a metric on a Riemannian space of constant curvature $\k$ given by \eqref{ccmetric} and both $\boldsymbol{g}$ and $\boldsymbol{F}$ have to satisfy  \eqref{EFE1}--\eqref{WC3}, \eqref{EFE4}, and \eqref{EFE5}.


\section*{Acknowledgments}

We are thankful to  M. Ortaggio and V. Pravda for useful comments on the draft.
AP acknowledges support from research plan RVO: 67985840 and research
grant GA\v CR 13-10042S. 


\appendix
\section{Curvature invariants in Kundt spacetimes of Weyl and traceless Ricci type N} \label{apendixB}
Let us study scalar curvature invariants in the Kundt subclass of Weyl type N spacetimes with the Ricci tensor of the form 
\begin{equation}\label{tracelessRicciN}
R_{ab} = \l g_{ab} + \eta k_a k_b,
\end{equation}
where $\l$ is a constant, $\eta$ is a scalar and $\Bk$ is a null vector. 
In general, a Kundt vector 
in a Ricci type II spacetime is necessarily a multiple WAND (see e.g. proposition 2 of \cite{ricciclanek}). 
Thus, in our case, it must coincide with the common AND $\Bell$ of the Weyl and the Ricci tensor. 
Hence, the spacetime is of Riemann type II (i.e. the Riemann tensor is of type II) with the Kundt AND $\Bell$. 
In addition to this, all boost weight zero null frame components of the Riemann tensor are constant (depending on $\l$). 

If $\l = 0$, the Riemann tensor is of type N and it is aligned with the Kundt vector $\Bell$. 
Thus, according to theorem 1 of \cite{HDVSI2} on characterization of $VSI$ spacetimes, the corresponding spacetime is $VSI$. 

If $\l$ is non-vanishing, then the spacetime is at least $CSI_0$, i.e. all curvature invariants constructed solely from the Riemann tensor (without incorporating its covariant derivatives) are constant.  Indeed, any full contraction of a tensor $\boldsymbol{T}$ is completely determined by its boost weight zero part  $\boldsymbol{T}_{(0)}$. Since in our case, the boost weight zero part $\boldsymbol{R}_{(0)}$ of the Riemann tensor $\boldsymbol{R}$ possesses only constant null frame components, also any full contraction of a tensor given by a series of tensor products of the Riemann tensor with itself or with the metric (which consists only of its boost weight zero part) is necessarily constant. See also section 2.3 of \cite{csiclanek}.

In order to study higher-order invariants constructed using also  covariant derivatives of $\boldsymbol{R}$, 
we make use of the balanced scalar approach introduced in \cite{VSI4d}, see also \cite{JiBiVPN}.  First, let us  define  $k$-balancedness of a tensor.
\begin{definition}We say that a tensor $\boldsymbol{T}$ is $k$-balanced if there exists a null vector $\Bell$ such that $\textnormal{bo}_{\Bell}(\boldsymbol{T}) < -k$   
and for any of its  null frame components $\eta$ with boost weight $b<-k$, the derivative $D^{-b-k} \eta$ is zero.    
\end{definition} 
\noindent Now,  let us  prove the following result on $k$-balancedness of covariant derivatives $\nabla^{(n)} \boldsymbol{R}$ of the Riemann tensor.  
\begin{propos}
In a Weyl type N Kundt spacetime with the Ricci tensor of the form \eqref{tracelessRicciN}, an arbitrary covariant derivative of the Riemann tensor is 1-balanced and thus  it is $VSI$.
\end{propos}
\begin{proof}
Let us stress out again that { such a} spacetime is necessarily aligned { \cite{NNspacetimes}}, and hence without loss of generality one can assume that $\Bk = \Bell$. Also, using {the} notation of the GHP formalism {\cite{ghpclanek}}, the boost weight $(-2)$ component of the Ricci tensor reads $\omegap = \eta$. 

Consider the Ricci decomposition of the Riemann tensor \cite{stephani}, 
\begin{equation}
\boldsymbol{R} = \boldsymbol{G} + \boldsymbol{E}+\boldsymbol{C},
\end{equation}
where $\boldsymbol{G}$ is the so-called scalar part ($G_{abcd} \propto R g_{a[c}g_{d]b}$
), 
$\boldsymbol{E}$ is the semi-traceless part ($E_{abcd} \propto g_{a[c}S_{d]b}-g_{b[c}S_{d]a}$,  
$S_{ab}$
being the traceless part of the Ricci tensor) 
and the Weyl tensor $\boldsymbol{C}$ is the
fully traceless part of the Riemann tensor. 
 
Since the Ricci scalar $R$ is constant ($R \propto \l$), we have $\nabla \boldsymbol{G} = \boldsymbol{0}$ and hence 
\begin{equation}\label{nablaR}
\nabla \boldsymbol{R} = \nabla \boldsymbol{C} + \nabla \boldsymbol{E}.
\end{equation}
For Weyl type N Kundt spacetimes with the Ricci tensor of the form (\ref{tracelessRicciN}), equations (28)--(30) in \cite{univIIIN}
are still valid together with $D\omega'=0$ and thus the following generalizations of lemmas 4.2 and 4.3 in \cite{univIIIN} are also 
valid\footnote{As in \cite{univIIIN}, we consider behaviour of scalars with respect to constant boosts.}
\begin{lemma}
\label{lem_1bal_coef}
In a Weyl type N Kundt spacetime with the Ricci tensor of the form \eqref{tracelessRicciN}, for a 1-balanced scalar $\eta$, scalars $L_{11} \eta$,  $\tau_i \eta$, $L_{1i}\eta$, $ \kappa'_i 
\eta $,  $ \rho'_{ij} 
\eta$,  $\Mi_{\!j1}\!  \eta$, $ \Mi_{\!kl}\!  \eta$ and $D \eta ,\ \delta_i \eta,\ \T \eta $ are also 1-balanced scalars.
\end{lemma}
\begin{lemma}
In a Weyl type N Kundt spacetime with the Ricci tensor of the form \eqref{tracelessRicciN}, a covariant derivative of a 1-balanced tensor
is again a 1-balanced tensor.
\end{lemma}
\noindent Since the Weyl tensor and the traceless part of the Ricci tensor, $\eta \ell_a\ell_b$, as well as $\boldsymbol{E}$ are 1-balanced, their arbitrary derivative
is also 1-balanced and so is any covariant derivative of the Riemann tensor, i.e. $\nabla^{(k)} \boldsymbol{R}$ is 1-balanced for any $k \in \mathbb{N}$.
\end{proof}

Thus, although the Riemann tensor is not $VSI$, its first covariant derivative is.  In particular, this means that the only non-trivial scalar curvature invariants are precisely those constructed from the Riemann tensor itself. However, we already know that all these invariants are necessarily constant. Therefore, we immediately obtain the following result. 
\begin{corol}\label{CSIKundt}
A Weyl type N Kundt spacetime with the Ricci tensor of the form \eqref{tracelessRicciN} is $CSI$.
\end{corol}


\nocite{*}
\bibliographystyle{unsrt}
\renewcommand\refname{References}
\bibliography{refs}  

\end{document}